\documentclass{aamas2010_cameraReady}

\usepackage{times}
\usepackage{dsfont}
\usepackage[all]{xy}
\usepackage{url}

\newtheorem{theorem}{Theorem}
\newtheorem{lemma}[theorem]{Lemma}

\newtheorem{corollary}[theorem]{Corollary}
\newtheorem{definition}[theorem]{Definition}
\newtheorem{example}[theorem]{Example}

\providecommand{\fullver}[1]{}
\fullver{\newcommand{\shortver}[1]{}}
\providecommand{\shortver}[1]{#1}

\newcommand{\nei}{\mathcal{N}}

\newcommand{\indeg}{\mathcal{I}}

\newcommand{\puregg}{\mbox{PURE-GG}}
\newcommand{\ppuregg}{\mbox{p-PURE-GG}}
\newcommand{\purechg}{\mbox{PURE-CHG}}
\newcommand{\ppurechg}{\mbox{p-PURE-CHG}}
\newcommand{\pureugg}{\mbox{PURE-UGG}}
\newcommand{\ppureugg}{\mbox{p-PURE-UGG}}

\newcommand{\hyper}{\mathcal{H}}
\newcommand{\dig}{\mathcal{D}}
\newcommand{\pri}{\textit{pri}}

\newcommand{\red}{\textit{red}}

\newcommand{\hide}[1]{}

\providecommand{\note}[1]{~\\\frame{\begin{minipage}[c]{\columnwidth}\vspace{2pt}\center{#1}\vspace{2pt}\end{minipage}}\vspace{3pt}\\}

\begin{document}
% In the original styles from ACM, you would have needed to
% add meta-info here. This is not necessary for AAMAS 2010 as
% the complete copyright information is generated by the cls-files.

% For appropriate information about authors and title written in the
% copyright information, you must use these commands. Note that the copyright
% box is not required for the initial submission.

\AuthorsForCitationInfo{Albert Xin Jiang and MohammadAli Safari}

\TitleForCitationInfo{Pure Nash Equilibria: Complete Characterization of
Hard and Easy Graphical Games}

\title{Pure Nash Equilibria: Complete Characterization of \\
Hard and Easy Graphical Games}

%\author{Tracking Number: 742}

\numberofauthors{2}
\author{
% You can go ahead and credit any number of authors here,
% e.g. one 'row of three' or two rows (consisting of one row of three
% and a second row of one, two or three).
%
% The command \alignauthor (no curly braces needed) should
% precede each author name, affiliation/snail-mail address and
% e-mail address. Additionally, tag each line of
% affiliation/address with \affaddr, and tag the
% e-mail address with \email.
\alignauthor
Albert Xin Jiang\\
       \affaddr{Department of Computer Science}\\
       \affaddr{University of British Columbia}\\
       \affaddr{Vancouver, Canada}\\
       \email{jiang@cs.ubc.ca}
\alignauthor
MohammadAli Safari\\
       \affaddr{Department of Computing Engineering}\\
       \affaddr{Sharif University of Technology}\\
       \affaddr{Tehran, Iran}\\
       \email{msafari@sharif.edu}
}

%\author{Albert Xin Jiang}
%\address{A.~X.~Jiang: Department of Computer Science,
%University of British Columbia,
%Vancouver, Canada.}
%\email{jiang@cs.ubc.ca}

%\author{MohammadAli Safari}
%\address{MohammadAli Safari: Department of Computing Engineering,
%Sharif University of Technology, Tehran, Iran.}
%\email{msafari@sharif.edu}

%\hide{
%\author{
%\begin{tabular}{cc}
%\begin{tabular}{c}
%Albert~Xin~Jiang  \\
%Department of Computer Science \\
%University of British Columbia \\
%Vancouver, BC. \\
%\url{jiang@cs.ubc.ca}
%\end{tabular}
%&
%\begin{tabular}{c}
%MohammadAli Safari \\
%Department of Computing Science \\
%University of Alberta \\
%Edmonton,  AB \\
%\url{msafarig@cs.ualberta.ca}
%\end{tabular}
%\end{tabular}
%}
%}

\maketitle

\begin{abstract}
We consider the computational complexity of pure Nash equilibria in graphical games.
It is known that the problem is NP-complete in general, but tractable (i.e., in P) for
special classes of graphs such as those with bounded treewidth.
It is then natural to ask: is it possible to characterize all tractable classes of graphs for this problem?
In this work, we provide such a characterization for the case of bounded in-degree graphs,
thereby resolving the gap between existing hardness and tractability results.
In particular, we analyze the complexity of $\puregg(C, -)$, the problem of deciding the existence of pure Nash equilibria in graphical games whose underlying graphs are restricted to class $C$.
We prove that, under reasonable complexity theoretic assumptions, for every recursively enumerable class $C$ of  directed graphs with bounded in-degree, $\puregg(C, -)$ is in polynomial time if and only if
the reduced graphs (the graphs resulting from iterated removal of sinks) of $C$ have bounded treewidth.
We also give a characterization for $\purechg(C,-)$, the problem of deciding the existence of pure Nash equilibria in
colored hypergraphical games, a game representation that can express the additional structure that some of the players
have identical local utility functions.
We show that the tractable classes of bounded-arity colored hypergraphical games are precisely those whose reduced graphs have bounded
treewidth modulo homomorphic equivalence.
Our proofs make novel use of Grohe's characterization of the complexity of homomorphism problems.
\end{abstract}


\section{Introduction}
%\note{I've tried to provide more motivation for our problem, give a more detailed summary of our results and approaches in the introduction}
There has been recent interest in using game theory to model and analyze large multi-agent systems
such as network routing, peer-to-peer file sharing, auctions and other market mechanisms.
%Many such games do not have closed form solutions for their solution concepts such as Nash equilibria, thus requiring us to find such solutions computationally.
One fundamental class of computational problems in game theory is the computation of \emph{solution concepts} of a finite game, such as Nash equilibria.
These kinds of computational tasks can be understood in the language of AI as \emph{reasoning} about the game: what are the likely outcomes of the game, under certain models of rationality of the agents? The goal is to be able to efficiently carry out such reasoning for real-world systems.

%nf: exponential
Much of the existing game theoretic literature models
simultaneous action games using the normal
form (also known as the strategic form), i.e. a game's payoff function is represented as a matrix with one entry
for each player's payoff each combination of the actions of all players.
The size of this representation grows exponentially in the number of players. Computations that are ``polynomial-time'' in the input size are nevertheless impractical. As a result the normal form is unsuitable for
representing large systems.
%\fullver{Although several computational tasks such as finding pure Nash equilibria and computing expected payoff under mixed strategies are polynomial time in the size of the normal form representation, they are intractable for large games because of the exponential size of the representation. }%end fullver

%compact reps
Fortunately, most real-world large games have highly-structured utility functions, which allows them to be represented compactly.
A line of research thus exists looking for \emph{compact game representations} that are able to succinctly describe structured games, and efficient algorithms for computing solution concepts that run in time polynomial  in the size of the representation.
An influential compact representation of games is \emph{graphical games} proposed by Kearns \emph{et al.}                                 \cite{graphical}.
A graphical game is associated with a graph whose nodes correspond to the players of the game and
edges correspond to payoff influence between players.
In other words, each player's payoffs depend only on the actions of himself and his neighbors in the graph.
%If the graph is sparse, each payoff function can be represented by a low-dimensional table.
The representation size of a graphical game is exponential in the size of its largest neighborhood. This can be exponentially smaller than
the normal form representation of the same game, especially for sparse graphs.

%other graphical rep: MAID; AGG
%Other notable graphical representations of games include
%Multi-agent influence diagrams (MAIDs) \cite{maid},
%which is an extension of influence diagrams to the multi-agent setting.
%Action graph games (AGGs) \cite{ActionGraph,JiangKLB06AGG} are a graphical representation that exploits the anonymity and context-specific independence structure of games.
%It is based on the action graph, a graph with actions as nodes. A player's utility for playing some action depends on the numbers of players that chose neighboring actions.

A compact game representation is not very useful if we cannot perform computations that are efficient relative to its size.
\fullver{
The classic computational game theory problem is to find a mixed-strategy Nash equilibrium, which always exist.
For graphical games on tree graphs, \cite{graphical} proposed a message-passing algorithm
(similar to probabilistic inference algorithms on graphical models) that computes approximate Nash equilibria in polynomial time.
%Elkind (path)
When the graph is a path, exact Nash equilibria can be found in polynomial time \cite{Elkind06}.

}%end fullver
\hide{Graphical games have also been instrumental in the recent proofs \cite{GoldbergPapa06,Daskalakis05} that eventually established the PPAD-completeness of the problem of finding a mixed-strategy Nash equilibrium in a normal form game of at least two players \cite{ChenDeng06}.
}
%PS NE
In this paper we focus on the problem of computing pure-strategy Nash equilibria (PSNE).
Unlike mixed-strategy Nash equilibria, which are guaranteed to exist for finite games \cite{Nash1951},
    in general pure Nash equilibria are not guaranteed to exist.
     Nevertheless, in many ways pure Nash equilibrium is a more attractive solution concept than mixed-strategy Nash equilibrium. First, pure Nash equilibrium can be easier to justify because it does not require the players to randomize. Second, it can be easier to analyze because of its discrete nature (see, e.g., \cite{Brandt}).
%related work on pure NE on GG
Gottlob \emph{et al.} \cite{Gottlob03} were the first to
analyze the problem of computing pure Nash equilibria in graphical games. They proved that the problem is NP-complete in general, even when the graphs have neighborhood size at most 3. On the other hand, for games with graphs of bounded hypertree-width there exists a dynamic-programming algorithm that determines the
existence of pure Nash equilibria in polynomial time in the size of the representation.
Daskalakis and Papadimitriou \cite{Daskalakis06} reduced the problem of finding pure strategy Nash equilibrium in graphical games to a Markov Random Field (MRF),
and then applied the standard clique tree algorithm to the resulting MRF. Among their results they showed that
for graphical games on graphs with log-sized treewidth, bounded neighborhood size and bounded number of actions per player,
deciding the existence of pure Nash equilibria is in polynomial time.

\fullver{\note{Maybe move or delete:

There have also been recent development on analyzing the probability of existence of PSNE in graphical games with random payoffs and/or random graph structure \cite{Daskalakis07random,Dilkina07}.
}
}%end fullver
%somewhat related:
%constrained PS

A natural question arises: are there other tractable classes of graphical games? Such a tractable class can be defined by restrictions over the graph structure as well as the local utility functions. In this paper, we analyze the complexity of $\puregg(C, -)$, the problem of determining the existence of pure Nash equilibria in graphical games whose underlying digraphs\footnote{We define graphical games on directed graphs
(whereas Daskalakis and Papadimitriou's  \cite{Daskalakis06} definition is based
on undirected graphs). The definition with directed graphs is more general, as graphical games on undirected graphs
can be thought of as games on directed graphs with bi-directional edges.
Our result applies to undirected graphs as a special case.} are restricted to class $C$ (while other aspects of the game representation can be arbitrary). We say $C$ is \emph{tractable} if $\puregg(C,-)$ is in polynomial time.

Throughout the paper we make the restriction that the graphical games have bounded neighborhood size (i.e.\ bounded in-degree).
%We think bounded-indegree graphical games are the most important kind, since the
%representation size of a graphical game is exponential in the maximum in-degree.
Graphical games with large in-degree
have the same problem as normal form games: even if we find polynomial-time algorithms for them,
that would still be impractical due tb the large input size.

Previous results \cite{Gottlob03,Daskalakis06} do not answer whether bounded tree-width is the sole measure of tractability of $\puregg(C,-)$. For example, it was unknown whether games with
log-sized treewidth and unbounded number of actions per player are tractable.
Furthermore, there are other graph parameters that affect the tractability of certain computational problems on directed graphs, e.g.
directed tree-width \cite{johnson2001directed}, D-width \cite{safari2005}, DAG-width \cite{obdrzalek2006dag}, and Kelly-width \cite{hunter2008digraph}. Since these parameters take advantage of the directionality of the edges, they could potentially give a better characterization of the tractability of $\puregg(C,-)$.
%\note{other examples of graph parameters?}

%Previously, the only known hardness result for subclasses of graphical games is
% Gottlob \emph{et al}'s NP-hardness result for the class of all graphs with neighborhood size at most 3.

%main result.
In this paper we give a  complete characterization of the tractable classes of bounded-indegree graphs, thereby resolving the gap between existing tractability and hardness results. Our results are summarized below.
\begin{enumerate}
\item Bounded-treewidth graphs are \emph{not} the only tractable kind of digraphs.
%\note{briefly discuss reduced graphs}
One example is graphical games on DAGs, for which pure equilibria always exist and can be computed efficiently.
More generally, whenever there is a sink (a vertex with out-degree zero), the utilities for that sink player do not affect the existence of pure equilibria.

\item Given a digraph $G$, let its \emph{reduced graph} %$\red(G)$
be the graph obtained by iterated removal of sinks. We prove that, under reasonable complexity theoretic assumptions, for every recursively enumerable class $C$ of directed graphs with bounded in-degree, $\puregg(C, -)$ is in polynomial time if and only if %$\red(C)$, 
the reduced graphs of $C$ have bounded treewidth.
\fullver{In other words,
if we only restrict the graphical structure, then
 the only tractable classes of graphical games are those with bounded treewidth after iterated removal of sinks.}
\item We consider \emph{colored hypergraphical games}, a new game representation based on colored hypergraphs, which can express the additional structure that some of the players have identical local utility functions. For the pure equilibrium problem on this representation $\purechg(C,-)$, we show that a class of colored hypergraphs is tractable if and only if its reduced graphs have bounded treewidth modulo homomorphic equivalence. This is a wider family of tractable games compared to the graphical game representation.
That is, by incorporating more information about the structure of the game into the graph, we are able to identify new tractable classes of games.
\end{enumerate}
Our results for $\puregg(C, -)$ follow as a corollary to our results for $\purechg(C, -)$.
Another corollary is that if the graphical games are represented as undirected graphs, then the tractable classes of undirected graphs are precisely those with bounded treewidth.

We prove these results by connecting $\puregg(C, -)$ and \\
$\purechg(C,-)$ to \emph{homomorphism problems}, which given colored hypergraphs $G$ and $H$, ask whether there exists a homomorphism from $G$ to $H$.
We then make use of Grohe's \cite{grohe} breakthrough result that characterizes the tractable classes of $\textit{HOM}(\mathcal{C},-)$,
homomorphism problems with restricted left-hand side.
This is (as far as we know) a novel application of Grohe's result for homomorphism problems to computational problems in game theory.
%\note{briefly discuss the proof technique?}
We prove our main tractability result by reducing an arbitrary instance of $\purechg(C,-)$ to an instance of the homomorphism problem. This reduction has a similar structure as \cite{Gottlob03}'s formulation of graphical games as constraint satisfaction problems.
\fullver{By adapting this type of constructions to the case of colored hypergraphical games, we are able to show that $\purechg(C,-)$ is tractable for a wider family of structures
(bounded treewidth modulo homomorphic equivalence) than $\puregg(C,-)$ (bounded treewidth).
%Such games can be solved by transforming to the corresponding homomorphism problems which have known polynomial-time algorithms \cite{dalmau,grohe}.
\note{an example of colored hypergraph game with small modulo-treewidth}
}%end fullver
On the other hand, our proof for our hardness result is quite unlike the existing NP-hardness proof for graphical games \cite{Gottlob03}.
At a high level, this is because the previous approach would construct graphical games on graphs with a certain specific structure. %take a homomorphism problem instance $(G,H)$ and create an
%equivalent graphical game on a graph that is different from $G$ in a specific way.
This is sufficient for proving NP-hardness, but not for our purposes, because we want to characterize
the complexity for $\purechg(C,-)$ for arbitrary $C$, which implies that we had to instead construct our graphical/colored hypergraphical game on an arbitrarily given digraph/colored hypergraph. In other words, we only have control over the utility functions (but not the graph structure), and need to set the utilities such that there is a solution to the given homomorphism problem if and only if the game has a pure equilibrium. This makes our task more technically challenging.
We think our proof techniques may have wider interest; for example, it might be possible to extend these techniques to prove similar results for action-graph games \cite{ActionGraph}, another compact game representation.

%what does this mean for real systems
These results
increase our understanding of the power and limitations of the graphical game representation, and have immediate practical impact.
Specifically, they imply that
if the systems we are interested in have large-treewidth reduced graphs when modeled as graphical games, then the resulting graphical games are unlikely to admit an polynomial-time algorithm for pure Nash equilibria,
\emph{even if the graphs have other types of structure}.
Nevertheless, if some of the players have identical utility functions, we might be able to get around this limitation of graphical games by representing the systems as colored graphical games instead. If the corresponding reduced graphs have bounded treewidth modulo homomorphism equivalence, pure Nash equilibria can be found efficiently by transforming to the corresponding homomorphism problems which have known polynomial-time algorithms
\cite{DBLP:conf/cp/DalmauKV02,grohe}.

%commented out because we only consider the decision problem --albert
%Our tractability result can
%also be straightforwardly applied to finding approximate mixed strategy Nash equilibria, since
%most of the algorithms proposed for solving approximate mixed strategy equilibria on graphical games \cite{graphical,VickreyKoller}
%can be thought of as discretizing the
%mixed-strategy space of the game, treating the discretized mixed strategies as pure strategies, and then solving for pure strategy equilibria
%of the resulting graphical game.

%\note{overview}
\section{Preliminaries}
\subsection{Graphical Games}

%\subsubsection*{Game}
%We first formally define a \emph{game}.
%\begin{definition}
A (simultaneous-move) \emph{game} is a tuple
$(N, \{S_i\}_{i\in N}, \{u_i\}_{i \in N})$
where
%\begin{itemize}
 $N=\{1,\ldots,n\}$ is the set of agents;
 for each agent $i$, $S_i$ is the set of $i$'s actions. $S_i$ is nonempty.
%We denote by $s_i\in S_i$ one of agent $i$'s actions.
An action profile $\mathbf{s} \in \prod_{i\in N}S_i$ is a tuple of actions of the $n$ agents.
 $u_i: \prod_{j\in N} S_j \rightarrow \mathds{R}$
is $i$'s utility function, which specifies $i$'s utility given any action profile.
%\end{itemize}
%\end{definition}

For every action profile $\mathbf{s}$, let $s_i$ be the action of agent $i$ under this action profile, and $s_{-i}$ be the $(n-1)$-tuple of the actions of
agents other than $i$ under this action profile.
For each action $s'_i \in S_i$, let $(s'_i, s_{-i})$  be the action profile where agent $i$ plays $s'_i$ and all the other agents play according to $s_{-i}$.

%\subsubsection*{Game Representations}
A game representation is a data structure that stores all information needed to specify a game.
%nf
\fullver{A \emph{normal form} representation of a game uses a matrix $U_i$ to represent each utility function $u_i$.
The size of this representation is $ n \prod_{j\in N} |S_j|$, which is $O(nm^n)$ where $m= \max_{i\in N} |S_i|$. }%end fullver

%gg

\begin{definition}
A \emph{graphical game} representation is a tuple\\
$(G, \{U_i\}_{i\in N} )$
where
\begin{itemize}
\item $G = (N,E)$ is a directed graph, with the set of vertices corresponding to the set of agents. $E$ is a set of ordered tuples corresponding to
the arcs of the graph,
i.e. $(i,j)\in E$ means there is an arc from $i$ to $j$.
Vertex $j$ is a \emph{neighbor} of $i$ if $(j,i)\in E$.
\item for each $i\in N$, a \emph{local utility function}
$U_i: \prod_{j\in \nei (i) } S_j \rightarrow \mathds{R}$
where $\nei(i)=\{i\}\cup \{ j \in N | (j,i)\in E\}$ is the \emph{neighborhood} of $i$.
\end{itemize}
\end{definition}

Each local utility function $U_i$ is represented as a matrix of size $\prod_{j\in \nei (i) } |S_j|$. Since the size of the local utility functions
dominates the size of the graph $G$, the total size of the representation is $O(n m ^ {(\indeg+1)} )$ where $\indeg$ is the maximum in-degree of $G$ and $m=\max_{j\in N}|S_j|$.
\hide{\begin{fact}
A class of graphical games have polynomial size in $n$ and $m$ if and only if the corresponding class of graphs has bounded in-degree.
\end{fact} }%end hide

A graphical game $(G,\{U_i\})$ specifies a game $(N, \{S_i\}, \{u_i\})$ where
each $S_i$ is specified by the domain of agent $i$ in $U_i$, and
for all $i\in N$ and all action profile $\mathbf{s}$ we have $u_i(\mathbf{s}) \equiv  U_i(s_{\nei(i)} )$, where
$s_{\nei(i)}=(s_j)_{j\in \nei(i)}$.

\subsection{Colored Hypergraphical Games}
We now consider graphical games with a certain additional structure. Specifically,
some players may have identical local utility functions.\footnote{For simplicity of presentation, we assume each player have the same number of actions.
We can convert an arbitrary game to our setting by adding dummy actions. Since we are only focusing on graphical games with
bounded in-degree, this would only increase the representation size by a polynomial factor.}

\fullver{Formally, we say two players $i, j$ have the same local utility function if $|\nei(i)|=|\nei(j)|$, and there exists a bijection
between their neighborhoods $\pi: \nei(i) \rightarrow \nei(j)$, such that for each local strategy profile $s_{\nei(i)}\in  \prod_{k\in \nei (i) }
S_k$, we have
\[
U_i(s_{\nei(i)}) = U_j(s^\pi_{\nei(j)})
\]
where $s^\pi_{\nei(j)}$ is the local strategy profile over $\nei(j)$ where $s^\pi_{\pi(k)} = s_k$ for all $k\in \nei(i)$.
}%end fullver

To represent this kind of  structure, we not only need to specify which players have the same local utility function, we also need to specify \fullver{the
bijections between neighborhoods, or equivalently} an ordering of the vertices in each neighborhood. We express this
kind of structure graphically using colored hypergraphs.

A \emph{colored hypergraph} $H = (V, E,C)$ consists of a set of vertices $V$, a set of edges $E$ where every edge $e\in E$ is
an ordered tuple of vertices in $V$,\footnote{Note that the definition we use is slightly different from the common definition of hypergraphs in which each edge is an unordered set of
vertices.}
and a color function $C: E\rightarrow\tau$ that maps each edge to its color.
In other words, each edge $e\in E$ is labeled with a color $C(e)$.
We denote as $V(H)$, $E(H)$ and $C_H$ the set of vertices, set of edges and
 the color function of colored hypergraph $H$, respectively.

We are now ready to define colored hypergraphical games.
%Intuitively, each player's neighborhood is represented by a colored hyperedge, and players with same-colored hyperedges have the same local utility function.
Intuitively, in a colored hypergraphical game, the players affecting player $i$'s utility are represented as a colored hyperedge consisting of these players' vertices, with $i$ being the first element. If two hyperedges have the same color, it means that their corresponding local utility functions are identical.
\begin{definition}\label{def:chg}
A \emph{colored hypergraphical game} is a tuple \\ $(G,\{U_c\}_{c\in\tau})$, where
\begin{itemize}
\item $G=(N,E,C)$ is a colored hypergraph with the set of colors $\tau$;
\item the set of vertices $V(G)=N$ corresponds to the set of players;
\item for each vertex $v \in N$, there exists exactly one edge $e\in E$ that has $v$ as the first element. Denote this edge as $e_v$.
\item for each color $c\in \tau$, edges of color $c$ have the same arity\footnote{The \emph{arity} of an edge $e$ is its size, i.e. number of elements.}
 $\mathcal{I}_c$.
\item each player has $m$ actions. Let $[m]=\{1,\ldots,m\}$.
\item for each color $c$, $U_c: [m]^{\mathcal{I}_c} \rightarrow R$.
\end{itemize}
\end{definition}

A colored hypergraphical game $(G, \{U_i\})$ specifies a game \\
$(N,\{S_i\},\{u_i\})$ where each $S_i = [m]$ and for each $i\in N$ and each action profile $\mathbf{s}$, $u_i(\mathbf{s})=U_{C(e_i)}(s_{e_i})$.

Unlike graphical games, where given an arbitrary digraph $G$ there is a graphical game on $G$, not all colored hypergraphs have corresponding colored hypergraph games. Let $\Sigma$ be the set of colored hypergraphs of colored hypergraph games. \fullver{From Definition \ref{def:chg}, it is straightforward to see that $\Sigma$ consists of colored hypergraphs $(V,E,C)$ such that for each vertex $v$, there exists exactly one edge $e\in E$ that has $v$ as the first element,
and furthermore edges of the same color have the same arity.
}%end fullver

Given $G\in \Sigma$, we define its \emph{induced digraph} $\dig(G)$ to be
a digraph on the same set of vertices; and for each hyperedge
$(v,v_1,\ldots, v_r)$ in $G$ we create directed edges
$(v_1,v),\ldots,  (v_r,v)$ in $\dig(G)$.

Graphical games can be thought of as special cases of colored hypergraph games where each neighborhood has a different local utility function, i.e. a different color.
Given a directed graph $G=(V,E)$, we define its \emph{induced colored hypergraph} $\hyper(G)=(V,\mathcal{E},C)$ such that its set of colors is $V$ and for each vertex $v\in V$, there is a hyperedge
$e\in \mathcal{E}$ of color $v$, consisting of vertices in $\nei(i)$, with $v$ being the first element in the tuple $e$.
The rest of the vertices in $e$ is sorted in a pre-determined order over $V$. In particular, if the vertices correspond to the agents
$1,\ldots,n$ in a game, we require these vertices to be sorted in ascending order of the agents.
By construction,  $\dig(\hyper(G))=G$ for all digraph $G$.
Given a graphical game $\Gamma=(G,\{U_i\}_{i\in N})$, its \emph{induced colored hypergraphical game} is $\hyper(\Gamma)=(\hyper(G), \{U_i\}_{i\in N})$. It is straightforward to verify that $\Gamma$ and $\hyper(\Gamma)$ represent the same game.

\fullver{When $G$ has bounded in-degree, then the size of individual edges in $\hyper(G)$ are bounded.}
For notational convenience, given a class of directed graphs $\mathcal{C}$, let $\hyper(\mathcal{C})$ be the class of induced hypergraphs of the directed
graphs in $\mathcal{C}$.

There is one graph  often associated with any hypergraph: the {\em
primal}  graph. Given a colored hypergraph $H$, its primal
graph $\pri(H)$ is an undirected graph obtained by making a clique out of the vertices in every edge
in $H$.
\fullver{For example, given a directed graph $G$ shown in Figure \ref{fig:digraph}, its induced colored hypergraph $\hyper(G)$ has the same set of vertices and the
hyperedges $\{A,B\}$, $\{A,B,C\}$, $\{D,E\}$, $\{C,D,E\}$,
$\{F,G\}$, $\{C,F,G\}$, and $\{B,C,D,E\}$. Figure \ref{fig:primal}
shows $\hyper(G)$'s primal graph.
}%end fullver

There are a couple of previous papers on the computational properties of
graphical games with different notions of identical utility functions.
Daskalakis \emph{et al.} \cite{Daskalakis05regular}
analyzed the complexity of finding pure and mixed Nash equilibria of graphical games on highly regular graphs (namely the $d$-dimensional grid), in which all local payoff functions are identical.
%They showed that finding pure Nash equilibria is tractable if $d=1$ and NEXP-complete otherwise.
Brandt \emph{et al.} \cite{DBLP:conf/wine/BrandtFH08}
instead analyzed graphical games on arbitrary graphs, but with several
stronger notions of symmetry.
In contrast, the colored hypergraphical game formulation places the least amount of restrictions and is thus more likely to occur in practice. In fact, these previous formulations can be thought of as special cases of colored hypergraphical games.

%\fullver{
\subsection{Best Response and Pure Nash Equilibrium}
Given a game $(N, \{S_i\}_{i\in N}, \{u_i\}_{i \in N})$ and $s_{-i}$, agent $i$'s \emph{best response} to $s_{-i}$ is
%\[
$BR_i(s_{-i}) = \arg\max_{s_i\in S_i} u_i(s_i,s_{-i})$.
%\]
Since $S_i$ is nonempty,  $i$ has at least one best response given any $s_{-i}$.
Note that in a graphical game, the best response of $i$ depends only on the actions of $i$'s neighbors.
\begin{definition}
An action profile $\mathbf{s} \in \prod_{i\in N}S_i$ is a \emph{pure Nash equilibrium} of the game $(N, \{S_i\}_{i\in N}, \{u_i\}_{i \in N})$
%if for each agent $i\in N$, and for each $s'_i \in S_i$,
%\[
%u_i (\mathbf{s}) \geq u_i(s'_i, s_{-i})
%\]
if each agent $i\in N$ is playing a best response to $s_{-i}$, i.e.
%\[
$s_i \in BR_i(s_{-i})$.
%\]
\end{definition}
%}%end fullver

\fullver{\subsection{Pure Nash Equilibria in Graphical Games} }
We define $\puregg$ to be the following decision problem: given a graphical game $(G, \{U_i\}_{i\in N}) $, decide whether there exists a pure Nash equilibrium.
Gottlob \emph{et al.}\ \cite{Gottlob03} have shown that this problem is NP-complete in general.
Given a class $C$ of digraphs and class $\mathcal{U}$ of local utility functions,  let $\puregg(C, \mathcal{U})$ be the pure Nash equilibrium decision problem
on graphical games when the graphs of the input game are taken only from class $C$ and the local utility functions are taken only from class $\mathcal{U}$.
In this paper we are interested in problems of the form $\puregg(C, -)$, which means that the local utility functions are unconstrained, other than the requirement
that the input is a well-formed graphical game, i.e. that each $U_i$ takes $|\nei(i)|$ arguments.

Similarly, we define the problem $\purechg(C,-)$ to be the pure Nash equilibrium decision problem
on colored hypergraphical games, with colored hypergraphs restricted to class $C$.

\fullver{We restrict our attention to classes of graphs with bounded in-degree. As discussed above, bounded in-degree is
a necessary and sufficient condition for the input graphical games to have polynomial size in $n$ and $m$.
Gottlob \emph{et al.}\ \cite{Gottlob03} showed that $\puregg(C, -)$ is NP-complete if we take $C$ to be the set of all
graphs with in-degree at most 3.
}%end fullver

\subsection{Treewidth}
\shortver{
Due to space constraints we omit the standard definition of treewidth for undirected graphs (see, e.g., \cite{Kloks94}).
The treewidth of a digraph $G$ is the treewidth of the undirected version of $G$. The treewidth of a colored hypergraph is the treewidth of its primal graph.
}
\fullver{One of the most significant recent advances in the field of algorithmics comes from the Graph Minors project of Robertson and Seymour.
In addition to being a major addition to the structure theory of graphs, the tools developed during their work imply algorithmic results
such as every minor-closed graph property can be decided in polynomial time.  The most far-reaching algorithmic contribution is the introduction
of graph decompositions such as tree decompositions and measures such as tree-width, which have helped identify large classes of tractable instances of hard (e.g. NP-complete)
graph problems.
}%end fullver

\fullver{
A {\em tree-decomposition} of an undirected graph $G = (V, E)$ is a
pair $(T, W)$, where $T$ is a tree, and $W$ is a function that assigns
to every node $i$ of $T$ a subset $W_i$ of vertices of $G$ such that
\begin{enumerate}
\item $\bigcup_{i\in T} W_i = V$,
\item For each edge $(u, v)\in E$, there exists some node $i$ of $T$
such that $\{u, v\}\subset W_i$, and
\item For all nodes $i, j, k$ in $T$, if $j$ is on the unique path
from $i$ to $k$ then $W_i \cap W_k \subset W_j$.
\end{enumerate}

The {\em width} of a tree-decomposition $(T, W)$ is the maximum of
$|W_i|-1$ over all nodes $i$ of $T$. The {\em tree-width} of $G$ is
the minimum width over all tree-decompositions of $G$.
Given a digraph $G$, the treewidth of $G$ is the treewidth of the undirected version of $G$.
}%end fullver

\hide{The key to the algorithmic success of tree decompositions is that they are readily extendable to arbitrary relational structures, such as
directed graphs and hypergraphs.  By considering tree
decompositions of the Gaifman (or primal) graph, large classes of tractable instances of hard problems can be found for various structures including hypergraphs
and directed graphs.\footnote{The main drawback of this approach is that often information is lost when considering the Gaifman graph, and this may be crucial.
}
}%end hide

\fullver{\subsection{Treewidth of Hypergraphs}
One question is what do we exactly mean by tree-width of a colored hypergraph.
There are various attempts for measuring the connectivity of hypergraphs. The two famous approaches are considering treewidth of the primal graph and {\em hypertree-width}, which was first introduced by Gottlob et al.~\cite{GLS99J}.
Daskalakis and Papadimitriou \cite{Daskalakis06} showed that when the size of individual edges are bounded both primal treewidth and hypertree-width are within a constant factor of each other. Hence, in this paper, we consider the primal tree-width as the tree-width of colored hypergraphs.
}%end fullver
\fullver{Figure
\ref{fig:treedecomp} shows a tree decomposition $(T,W)$ of the primal graph
 from Figure \ref{fig:primal}. Each node $i\in T$ of the tree is
labeled with $W_i$.
}%end fullver

\fullver{
\begin{figure*}[tb] \center
  \begin{minipage}[b]{.24\textwidth}
  \entrymodifiers={++[o][F-]} \SelectTips{cm}{}

  \centerline{
  \xymatrix @-1pc {
  %*\txt{ } & *\txt{ } & A\ar@(ur,dr)[] \ar@{<->}[d] & *\txt{ } & *\txt{ }\\
  %*\txt{ }& *\txt{ }& B\ar@(ur,dr)[]\ar@{<->}[d]& *\txt{ }& *\txt{ }\\
  *\txt{ }&A \ar@{<->}[r] & B\ar@{<->}[d]& *\txt{ }& *\txt{ }\\
  E\ar@{<->}[r]&D\ar@{<->}[r]&C\ar@{<->}[r]&F\ar@{<->}[r]&G
  }}

  \caption{A graph $G$.}\label{fig:digraph}
  %\end{figure}
  %\begin{figure}[htb]
  \end{minipage}
\hfill
  \begin{minipage}[b]{.24\textwidth}
  \entrymodifiers={++[o][F-]}
  \centerline{
  \xymatrix @-1pc {
  %*\txt{ } & *\txt{ } & A \ar@{-}[d] \ar@/^1pc/@{-} [dd]& *\txt{ } & *\txt{ }\\
  %*\txt{ }& *\txt{ }& B\ar@{-}[d] \ar@{-}[dr] \ar@{-}[dl]& *\txt{ }& *\txt{ }\\
  *\txt{ }& A \ar@{-}[r] \ar@{-} [dr]& B\ar@{-}[d] \ar@{-}[dr] \ar@{-}[dl]& *\txt{ }& *\txt{ }\\
  E\ar@{-}[r]\ar@{-}@/_1pc/[rr]
  &D\ar@{-}[r]\ar@{-}@/_1pc/[rr]&C\ar@{-}[r]\ar@{-}@/_1pc/[rr]&F\ar@{-}[r]&G
  } } \caption{$\pri(\hyper(G))$.}\label{fig:primal}
  %\end{figure}
  %\begin{figure}[tb]
  \end{minipage}
\hfill
  \begin{minipage}[b]{.34\textwidth}
  \entrymodifiers={+[F-]}

  \centerline{
  \def\objectstyle{\scriptstyle}
  \xymatrix @-1pc { *\txt{ } & {W_1= \{A,B,C\}} \ar@{-}[d] & *\txt{ }\\
  {W_3=\{C,D,E\}} \ar@{-}[r]& {W_2=\{ B,C,D,F\}} \ar@{-}[r] &
  {W_4=\{C,F,G\}} } }

  \caption{Tree decomposition of Figure \ref{fig:primal}.}\label{fig:treedecomp}
  \end{minipage}

\end{figure*}

}%end fullver

\fullver{
The following is a well-known property of tree decompositions.
\begin{lemma}[Kloks \cite{Kloks94}]\label{lemma:kloks}
If $X$ is a clique in an undirected graph $G$, then in any tree decomposition $(T,W)$ of $G$,  $\exists i\in T$ such that
$X\subseteq W_i$.
\end{lemma}
}%end fullver

\fullver{
\section{Homomorphism Problems}
In this section we give a brief introduction of homomorphism, homomorphism problems, and Grohe's characterization of tractable classes of homomorphism problems.
%\subsection{Graph Theory and Tree-width}
}%end fullver

\subsection{Homomorphism}
Let $G$ and $H$ be two colored hypergraphs. %Furthremore, assume that edges in both $G$ and $H$ have colors (not necessarily distinct colors).
A \emph{homomorphism} from $G$ to $H$ is  a mapping $h$ from the vertex set of $G$ to the vertex
set of $H$ that preserves both adjacency and color,
i.e. for every edge $e=(a_1, a_2, \cdots, a_k) \in E_G$, $h(e) = (h(a_1), h(a_2), \cdots, h(a_k)) \in E_H$ and  $C_G(e)=C_H(h(e))$. In a \emph{homomorphism problem}, we are given $G$ and $H$ and
have to decide whether there exists a homomorphism from $G$ to $H$. %Notice that we don't need the mapping function $h$ to be a bijection, however, it obviously cannot map two vertices of $G$ that share an edge to a single vertex in $H$.
%When the colors of edges are not explicitly specified (e.g. when $G $ and $H$ are hypergraphs or graphs), we assume that all edges in both $G$ and $H$ have the same color and, hence, any mapping is required to only preserve the adjacencies.

For two classes $\mathcal{C}$ and $\mathcal{D}$ of colored hypergraphs let $\textit{HOM}(\mathcal{C}, \mathcal{D})$ be the homomorphism problem when the input colored hypergraphs are taken only from classes $\mathcal{C}$ and $\mathcal{D}$.  When an input class is the class of all colored hypergraphs, we use the notation `$-$' instead.

\fullver{
Several well-known NP-complete problems are restrictions of the homomorphism problem. For example, deciding whether an undirected graph $G$ has a clique of size $k$ is equivalent to solving the homomorphism problem instance
%$\textit{HOM}(K_k, G)$
$(K_k,G)$ where $K_k$ is a clique of size $k$.
%Albert: I'm commenting out the following to avoid confusion of vertex-colors of 3-colorability with the edge-colors
%        we're using in this paper.
and deciding whether a graph $G$ is 3-colorable is equivalent to
solving the homomorphism instance
%$\textit{HOM}(G, H)$
$(G,K_3)$.
}%end fullver

\fullver{The homomorphism  problem is a fundamental problem and is of interest in various areas of computer science. Constraint satisfaction problems in Artificial Intelligence can be phrased as instances of the homomorphism problem\cite{298498}. Both the evaluation and    containment problems for conjunctive database queries are also equivalent to the homomorphism problem\cite{803397}.
}

Two hypergraphs $G$ and $H$ are \emph{homomorphically equivalent} if there is a homomorphism from $G$ to $H$ and vice versa. A class $C$ has bounded treewidth modulo homomorphic equivalence if there exists some constant $k$ such that every hypergraph in $C$ is homomorphically equivalent to a hypergraph with treewidth at most $k$. For example the class of bipartite graphs have bounded treewidth modulo homomorphic equivalence as they are homomorphically equivalent to an edge. We use $\textit{modulo-treewidth}(G)$ to indicate the minimum $k$ for which $G$ is homomorphically equivalent to a hypergraph of treewidth $k$.

\begin{example}\label{ex:modtw}
We describe a class of colored hypergraphical games with bounded treewidth modulo homomorphism equivalence.
Each game has $m^2+2m$ players. There are 4 colors $\{L,R,X,Y\}$.
We have $m$ players %with color $L$,
labeled $l_1\ldots l_m$, and $m$ players %with color $R$,
labeled \\
$r_1,\ldots,r_m$.
For each $i,j\in \{1,\ldots,m\}$ we have a player $x_{ij}$ %of color $X$
and a player $y_{ij}$. %of color $Y$.
For each player $x_{ij}$, we have an hyperedge $(x_{ij}, y_{ij}, l_i, r_j)$ of color $X$; for each player $y_{ij}$, we have an hyperedge $( y_{ij},x_{ij}, l_i, r_j)$ of color $Y$.
Furthermore for each player $l_i$ we have a hyperedge $(l_i)$ of color $L$ and for each player $r_i$ we have a hyperedge $(r_i)$ of color $R$.
The colored hypergraph is homomrophically equivalent to the fragment involving only the vertices $l_1, r_1, x_{11},y_{11}$ and their corresponding hyperedges.
Therefore these colored hypergraphical games have modulo-treewidth 3, while the treewidth of each hypergraph is at least $m$.
\end{example}

\subsection{Parameterized Complexity Theory}
Our results make use of certain concepts from he theory of parameterized complexity developed by Downey and Fellows~\cite{downey99parameterized}. They are not essential for understanding our reductions. We briefly mention the relevant concepts here and refer the reader to \cite{downey99parameterized,Grohebook} for more details.

%There are many algorithms whose running time is exponential, but is exponential in terms of some parameters that are expected to be small in typical input instances. For handling  such cases, the theory of parameterized complexity was developed by Downey and Fellows~\cite{downey99parameterized} to study such problems.

Given a decision problem $P\subseteq \Sigma^*$,  a parameterization of $P$ is a mapping $k:\Sigma^*\rightarrow N$ that defines the  paramterized problem $(P, k)$.
A paramterized problem $(P, k)$ is {\em fixed parameter tractable} if there is a computable function $f:N\rightarrow N$ and an algorithm that decides if a given instance $x\in \Sigma^*$ belongs to $P$ in time $f(k(x))|x|^{O(1)}$ for some function $f$ depending only on  $k$.
\fullver{For example, there is an algorithm~\cite{509224} which solves the vertex cover problem in $O(kn+1.271^k)$ time, where $n$ is the number of vertices and $k$ is the size of the vertex cover. Hence,  vertex cover is fixed-parameter tractable with respect to this parameter.}%end fullver
The class of all fixed parameter tractable problems is denoted by FPT.
\fullver{By defining fpt-reduction, as follows, between paramterized problems hardness and completeness of fixed parameter tractable problems are defined in the usual way.

An {\em fpt-reduction} from a  paramterized problem $(P, k)$ over $\Sigma$ to a paramterized problem $(P', k')$
over  $\Sigma_2$ is a mapping $R:\sigma^* \rightarrow \Sigma'^*$ such that:
\begin{itemize}
\item For all $x\in \Sigma^*$, $x\in P$ iff $R(x)\in P'$.
\item There is a computable function $f:N\rightarrow N$ such that $R(x)$ is computable in time $f(k(x)).|x|^c$ for some constant $c$.
\item There is a computable function $g:N\rightarrow N$ such that for all instances $x\in \Sigma^*$ we have $k'(R(x))\leq g(k(x))$.
\end{itemize}
The intuition behind the reduction is that if there is a fpt-reduction from a paramterized problem $(P, k)$ to a fixed parameter tractable problem $(P', k')$ then $(P, k)$ is fixed parameter tractable.
}%end fullver

Downey and Fellows~\cite{Downey1} defined the paramterized complexity class $W[1]$, which can be seen as an analogue of NP in parametrized complexity theory, and conjectured that FPT is a proper subset of W[1]. This conjecture is widely believed to be true.
%They also showed that the p-CLIQUE problem is $W[1]$-complete. In the p-CLIQUE problem the parameter is $k$ and we want to see if the input graph has a clique of size $k$.

%define parameterized version of the problems:
Let $\textit{p-HOM}(C, -)$ be the parametrized version
of $\textit{HOM}(C, -)$, with the parameter being the representation size of the left colored hypergraph.
Similarly we define  $\ppuregg(C,-)$ and  \\
$\ppurechg(C,-)$
to be the parametrized versions of \\
$\puregg(C,-)$ and $\purechg(C,-)$, with the parameters being the representation sizes of the directed graph and the colored hypergraph, respectively.
\subsection{Complexity of Homomorphism Problems}
Grohe~\cite{grohe}, in a breakthrough paper, characterizes the tractable instances of the homomorphism problem when we restrict the left input graphs.\footnote{Grohe stated his result on \emph{relational structures} instead of colored hypergraphs.
The two formulations are equivalent.}
\begin{theorem}[Grohe \cite{grohe}] \label{thm:grohe}
Assume $\textit{FPT}\neq \textit{W[1]}$. Then for every recursively enumerable class $C$ of colored hypergraphs with bounded arity the following  statements are equivalent.
\begin{enumerate}
\item $\textit{HOM}(C, -)$ is in polynomial time.
\item $\textit{p-HOM}(C, -)$ is fixed-parameter tractable.
\item $C$ has bounded modulo-treewidth.
\end{enumerate}

\end{theorem}
Under a slightly stronger assumption of $\textit{nonuniform-FPT} \neq$ \\
$\textit{nonuniform-W[1]}$, this result holds
for arbitrary (not necessarily recursively enumerable) class $C$.

%Notice that we are using colored hypergraphs instead of \emph{relational structures} used by Grohe. The two are equivalent and we find hypergraphs much simpler and readable to use.

\section{Main Result}\label{sec:main}
%\note{I've rearranged the subsections somewhat, moved some proofs to appendices. The subsection on the hardness for colored hypergraphical games is new.}

\subsection{Digraphs with sinks}
One's first intuition is to try and show a correspondence between
$\puregg(C, -)$ and $HOM(\hyper(C), - ) $ for  arbitrary classes of bounded-degree graphs.
In fact such correspondence does not exist for arbitrary graphs.
For example, a graphical game on a directed acyclic graph (DAG) always has a PSNE, which can be computed efficiently by a greedy algorithm that goes through vertices in the topological order, from sources to sinks.
Consider the class $D_k$ of $k$-bounded in-degree DAGs.
$\hyper(D_k)$ has unbounded modulo-treewidth, however $\puregg(D_k, -)$ is in polynomial time.
%Before going through the details of our result, let's make a further restriction on the  class of graphical games that we consider.
This is just an example of a more general phenomenon in graphical games. Let $GG = (G, \{U_i\}_{i\in N} )$ be a graphical game.
If $G$ has a sink $u$ (i.e. $u$ has out-degree zero) then the action of $u$ does not affect any other player. This means we can simply solve the game without player $u$ and the resulting game has a pure Nash equilibrium if and only if $GG$ has one.
\fullver{On the other hand, in homomorphism problems, a vertex with out-degree zero in the left hand graph has no special property.}

Intuitively, this is because
in a graphical game (and any game in general) each player has at least one action, and as a result, whatever actions
others chose, each player has at least one best response. \fullver{If we think of computing pure Nash equilibria in graphical games as a constraint satisfaction
problem,
A vertex with out-degree zero in a graphical game corresponds to a constraint that can always be satisfied.
Homomorphism problems (without restricting
the right hand side) do not have this property.}%end fullver

We formalize this intuition as the following  classification of digraphs into reducible and irreducible graphs:
\begin{definition}
A directed graph $G$ is \emph{irreducible} if it does not have a sink (a vertex with out-degree zero). Otherwise $G$ is \emph{reducible}.
%for every vertex $i\in V(G)$, there exists a strongly connected component (SCC) $\pi(i)$ with size at least 2 such that there is a directed path from $i$ to $\pi(i)$.
\end{definition}

It is often helpful to consider the strongly connected components (SCCs) of a directed graph. In particular, we can characterize irreducible graphs by their \emph{terminal SCCs}.
\begin{definition}
A strongly connected component(SCC) $\pi$ of $G$ is \emph{terminal} if there is no outgoing edges from $\pi$. A terminal SCC is by definition a maximal SCC.
\end{definition}

\begin{lemma}
If $G$ is irreducible then all its terminal SCCs have size at least 2.
\end{lemma}
This is because otherwise, a vertex in a terminal SCC with only one vertex would have out-degree zero.

It turns out that for our purposes, given an arbitrary digraph we can focus on its subgraph resulting from iterative removal of sinks.
\begin{definition}
Given a directed graph $G$, its reduced graph \red(G) is the result of the following algorithm:
\begin{enumerate}
\item repeat until $G$ does not change:
\begin{enumerate}\item remove all vertices with out-degree zero as well as their incoming edges.
\end{enumerate}
\item return $G$
\end{enumerate}
\end{definition}

\fullver{We define the \emph{reducible vertices} of $G$ to be the vertices removed by the above algorithm, i.e. vertices in $V(G)\setminus V(\red(G))$.
It is straightforward to see that $\red(G)$ is an irreducible graph, and every reducible vertex is a SCC of size 1 in $G$.

Consider the DAG over SCCs of $G$. A reducible vertex either
\begin{itemize}
\item has out-degree zero (and is thus a terminal SCC), or
\item all of its outgoing edges point to reducible vertices that correspond to SCCs higher in the topological ordering of the DAG over SCCs.
\end{itemize}
Furthermore any vertex that satisfies one of the above condition would be removed by the algorithm and is thus a reducible vertex.
}%end fullver
%So the set of reducible vertices and thus \red(G)
%don't depend on the order of removal in the above algorithm.

%--------------------------
%\begin{definition}
%A vertex $v$ in a graph $G$ is \emph{reducible} if  either
%\item it has out-degree zero, or
%\item all of its outgoing edges point to reducible vertices.
%\end{definition}

\begin{definition}
Given a graphical game $GG=(G, \{U_i\}_{i\in N})$, its \emph{reduced game} $\red(GG)$ is  $(\red(G),  \{U_i\}_{i\in V(\red(G))})$,
 i.e.\ the game obtained by removing all agents corresponding to reducible vertices of $G$.
\end{definition}
$\red(GG)$ is well defined because for all $v \in \red(G)$, vertices that are neighbors of $i$ in $G$ are not reducible vertices, so they are still present in $\red(G)$.

\begin{lemma}
A graphical game $GG$ has a pure Nash equilibrium if and only if its reduced game $\red(GG)$ has a pure Nash equilibrium.
\end{lemma}
\fullver{\begin{proof}
It is straightforward to see that if $GG$ has a pure Nash equilibrium $\mathbf{s}$, then $s_{\red(G)}$, the tuple of actions from agents in $\red(G)$,  is a pure Nash equilibrium of $\red(GG)$.

Now suppose $\red(GG)$ has a pure Nash equilibrium $s_{\red(G)}$. We now extend it to a pure Nash equilibrium of $GG$ by assigning actions to the reducible vertices of $G$. Consider the graph $H$ of the reducible vertices of $G$ and edges between them. $H$ is a DAG since every reducible vertex is a SCC.
Consider the vertices in $H$ in topological order and for each $i\in V(H)$, assign an  action that is a best response  to her neighbors' actions. This is possible because the neighbors of $i$ is either in $\red(G)$ or earlier in topological order in $H$.
\end{proof}
Note that the above can be carried out in polynomial time. This implies the following result.
}%end fullver

%This means that given a class $\mathcal{C}$ of graphs with bounded in-degree, if all induced hypergraphs of reduced graphs of  $\mathcal{C}$
%have bounded modulo-treewidth, then $\puregg(\mathcal{C},-)$ is in P.

%Are these all the tractable cases? In fact the answer is no, due to graphs that are exponentially larger than their reduced graphs.
%Formally, let $\mathcal{C}$ be a class of graphs whose reduced graphs have $k\log_m(n)$ vertices where $m$ is the maximum number of actions. Then a polynomial time algorithm for $\puregg(\mathcal{C}, -)$ is the following: given a graphical game with $n$ agents, compute its reduced graph, then solve for pure NE on the reduced game using brute force (i.e. treating it as normal form). The running time is $O(m^{(k\log_m(n))} ) = O(n^k)$.

%tractability
\begin{lemma}\label{lem:redC2C}
Suppose $\mathcal{C}$ is a recursively enumerable class of graphs with bounded in-degree,  such that $\puregg(\red(\mathcal{C}) , -)$ is in P.
Then $\puregg(\mathcal{C},-)$ is in P.
\end{lemma}

For the other direction, we would like to prove that
if graphical games on a class of graphs $\mathcal{C}$ is tractable, then graphical games on $\red(\mathcal{C})$ is also tractable.
This is not trivial, because although the reducible vertices of graphs in $\mathcal{C}$ do not affect the existence of pure equilibria, the subgraphs on these vertices could potentially carry information (similar to advice strings in complexity theory) such that $\puregg(\mathcal{C},-)$ is easier than $\puregg(\red(\mathcal{C}) , -).$
It turns out that if we consider the parameterized version of the problem, then if $\ppuregg(\mathcal{C},-)$ is in FPT then $\ppuregg(\red(\mathcal{C}),-)$ is in FPT. This will be sufficient for our purposes. The proof of the following lemma is given in Appendix \ref{sec:fptproof}.
\begin{lemma}\label{lem:fpt}
Suppose $\mathcal{C}$ is a recursively enumerable class of graphs with bounded in-degree,  such that $\ppuregg(\mathcal{C},-)$ is in FPT. Then $\ppuregg(\red(\mathcal{C}),-)$ is in FPT.
\end{lemma}

We can define analogous concepts for colored hypergraphs and colored hypergraphical games, by looking at their induced digraphs. A colored hypergraph $G\in \Sigma$ is irreducible if and only if its induced digraph is irreducible. Given $G\in \Sigma$, its reduced colored hypergraph $\red(G)$ is obtained by removing all reducible vertices of the induced digraph of $G$ and all hyperedges that include these reducible vertices. Reduced colored hypergraphical games can be defined similarly. Lemmas \ref{lem:redC2C} and \ref{lem:fpt} can be straightforwardly extended to colored hypergraphical games.

\subsection{Main Theorems}
The above implies that for the complexity of $\puregg(C,-)$ and
$\purechg(C,-)$, it is sufficient to consider irreducible graphs.
The complexity for a general class $C$ then correspond to the complexity for $\red(C)$.
It turns out that if we restrict to irreducible graphs, there exists a correspondence between $\purechg(C, -)$ and $HOM(\hyper(C), - ) $.
We are now ready to state our main result, which will be proved in the rest of Section \ref{sec:main}:

\begin{theorem}\label{thm:complexity}
Assume $\textit{FPT}\neq \textit{W[1]}$. Then for every recursively enumerable class of bounded arity colored hypergraphs $\mathcal{C}\subseteq \Sigma$,
the following statements are equivalent.
\begin{enumerate}
\item $\purechg(C,-)$ is in polynomial time.
\item $\ppurechg(C,-)$ is fixed-parameter tractable.
\item for every $G\in \mathcal{C}$, $\red(G)$ has bounded modulo-treewidth.
\end{enumerate}
\end{theorem}
The direction $1\rightarrow 2$ is trivial; the ``tractability'' direction $3\rightarrow 1$ is proved in Section \ref{sec:tractable}; the ``hardness" direction $2\rightarrow 3$ is proved in Sections
\ref{sec:hardness} and \ref{sec:chghardness}.

We then obtain as a corollary the characterization for the
complexity of $\puregg(C, -)$.
We make use of the following lemma on the modulo-treewidth of $\hyper(G)$.
\begin{lemma}\label{lem:modtw}
Given a digraph $G$, the modulo-treewidth of $\hyper(G)$ equals the treewidth of $\hyper(G)$.
\end{lemma}
\fullver{
\begin{proof}
Recall that a core of a relational structure/colored hypergraph $A$ is a minimal
substructure of $A$ that is homomorphically equivalent to $A$.
Lemma 2.5 of \cite{grohe} shows that the modulo-treewidth of $\hyper(G)$ equals the treewidth of the core of $\hyper(G)$.
We claim that the core of $\hyper(G)$
must equal $\hyper(G)$. Suppose otherwise, then
the core either has less vertices or less hyperedges or both.
If the core has less vertices, then the hyperedges corresponding to those vertices are also gone; since each vertex in $\hyper(G)$ has a corresponding hyperedge, so the core must have less hyperedges.
Since each hyperedge in $\hyper(G)$ has a unique color,
the core has less colors than $\hyper(G)$. Since in a homomorphism from $\hyper(G)$ to the core, a hyperedge of  color $c$ in $\hyper(G)$ must map to a hyperedge of the same color $c$ in the core, therefore there exists no homomorphism from $\hyper(G)$ to the core. This contradicts with the fact that the core is homomorphically equivalent to  $\hyper(G)$.
Therefore the modulo-treewidth of $\hyper(G)$ is equal to the
treewidth of the core, which is $\hyper(G)$.
\end{proof}
}%end fullver

Furthermore, we can relate  the treewidth of a digraph $G$ to the treewidth of $\hyper(G)$. Daskalakis and Papadimitriou \cite{Daskalakis06} showed that given an undirected graph $H$ with bounded degree, the treewidth of its induced hypergraph $\hyper(H)$ and the treewidth of $H$ are within a constant factor of each other.
This result cannot be directly applied to digraphs, because the
induced hypergraph of the undirected version of a digraph $G$ can be different from $\hyper(G)$. Nevertheless, their proof can be relatively straightforwardly adapted to digraphs, yielding the following lemma.
\begin{lemma}\label{lem:hypertw}
Given a digraph $G$ with bounded in-degree, the treewidth of $\hyper(G)$ and the treewidth of the undirected version of $G$ are within a constant factor of each other.
\end{lemma}
This means for our purposes bounded treewidth of $\hyper(G)$
implies bounded treewidth of $G$ and vice versa.
We are now ready to state the characterization for $\puregg(C, -)$,
which is a direct consequence of
Theorem \ref{thm:complexity} and Lemmas \ref{lem:modtw} and \ref{lem:hypertw}.
\begin{corollary}\label{cor:complexityGG}
Assume $\textit{FPT}\neq \textit{W[1]}$. Then for every recursively enumerable class $C$ of  digraphs with bounded in-degree the following statements are equivalent.
\begin{enumerate}
\item $\puregg(C, -)$ is in polynomial time.
\item $\ppuregg(C, -)$ is fixed-parameter tractable.
\item %$C$ has bounded modulo-treewidth, i.e.\
for every $G\in C$, $\red(G)$ has bounded treewidth.
\end{enumerate}
\end{corollary}

Comparing Theorem \ref{thm:complexity} and Corollary \ref{cor:complexityGG}, CHGs gives a wider family of tractable games compared to graphical games. For example, the class of CHGs described in Example \ref{ex:modtw} has  bounded modulo-treewidth  but unbounded treewidth. Thus they would be intractable if represented as graphical games.

We also obtain as a corollary the characterization for the pure equilibrium problem
for graphical games defined on undirected graphs. Define
$\pureugg(C, -)$ to be the problem of deciding existence of pure equilibrium on such undirected graphical games, restricted to the class of graphs $C$. \shortver{Then under the same assumptions, $C$ is tractable if and only if its graphs have bounded treewidth.}

\fullver{\begin{corollary}
Assume $\textit{FPT}\neq \textit{W[1]}$. Then for every recursively enumerable class $C$ of  undirected graphs with bounded degree the following statements are equivalent.
\begin{enumerate}
\item $\pureugg(C, -)$ is in polynomial time.
\item $\ppureugg(C, -)$ is fixed-parameter tractable.
\item %$C$ has bounded modulo-treewidth, i.e.\
for every $G\in C$, $G$ has bounded treewidth.
\end{enumerate}
\end{corollary}
}%end fullver

\subsection{Proof of Tractability Result}\label{sec:tractable}

%\note{I've changed the statement and proof to colored version.
%The main difference in the proof is that $H'$ now only has $m$ vertices.}
\fullver{
In this section we prove direction $3\rightarrow 1$, restated as follows:
\begin{theorem}\label{thm:tractable}
Consider a recursively enumerable class $C \subset \Sigma$ of colored hypergraphs with bounded arity.
If for every $G\in C$, $\red(G)$ has bounded modulo-treewidth,
then $\textit{PURE-CHG}(C, -)$ is in polynomial time.
\end{theorem}
}%end fullver

We use the following lemma that reduces a colored hypergraphical game to a homomorphism problem instance.
\shortver{The tractability direction of Theorem \ref{thm:complexity} then follows straightforwardly.}
\begin{lemma}\label{lem:tractable}
Let $\Gamma = (G, \{U_i\}_{i\in N} )$ be a colored hypergraphical game.
It is possible to construct in polynomial time an instance $(G',H')$of homomorphism problem
such that $\Gamma$ has a pure equilibrium if and only if
there exists a homomorphism from $G'$ to $H'$.
Furthermore if $G$ has bounded arity and bounded modulo-treewidth
then so does $G'$.
\end{lemma}
\begin{proof}
Given  a colored hypergraphical game $\Gamma = (G, \{U_c\}_{c\in \tau} )$,
each player having $m$ actions, we construct the instance $(G', H')$ of the homomorphism problem as follows.
%\begin{description}
%\item[$G'$:]
Let $G' = G$.
%\item[$H'$:]
%Create a vertex $u^a$ of $H'$ for every action $a$ of every agent $u$ in $GG$.
$H'$ consists of  $m$ vertices, one for each action in $[m]$.
For each color $c$, for each action tuple $(a,a_1, a_2, \cdots, a_{r})$
such that
$$a\in \arg\max_{a'\in[m]} U_c(a', a_1, a_2, \cdots, a_{r})$$
(i.e. $a$ is a best response for a player with utility function $U_c$
given neighbor actions $(a_1, a_2, \cdots, a_{r})$),
create an hyperedge\\
 $(a,a_1, a_2, \cdots, a_{r})$ of $H'$ with color $c$.

%Corresponding to every edge $e^u = (u, u_1,u_2, \cdots, u_{r})$ in $G'$ and every acceptable action tuple $(a, a_1, a_2, \cdots, a_{r})$ for players $u, u_1, \cdots, u_{r}$ create an edge $$(u^a, u^{a_1}_1, u^{a_2}_2, \cdots, u^{a_{r}}_{r})
%$$ with the same color as $e^u$. $(a, a_1, a_2, \cdots, a_{r})$ is an  acceptable tuple if $a$ is a best response action for player $u$
%when  her neighbors are playing actions $(a_1, a_2, \cdots, a_{r})$.

%\end{description}

If $\Gamma$ has a pure Nash equilibrium then the mapping that maps each vertex $u$ to the vertex $a$, where $a$ is the action chosen by $u$ in the pure Nash equilibrium, is a homomorphism. For the other direction, if $H'$ is a homomorphism of $G'$ and the corresponding mapping function is $\ell$ then
%\begin{itemize}
 $\ell(u)$ corresponds to an action of $u$,
and for every edge $e^u = (u, u_1, u_2, \cdots, u_{r})$ of color $c$ in $G'$,  $\ell(e^u) =  (\ell(u), \ell(u_1), \ell(u_2), \cdots, \ell(u_{r}))$  must be an edge of color $c$ in $H'$.
This implies that  $\ell(u)$ is a best response action of player $u$ against his neighbors' actions.
%\end{itemize}
Therefore, the mapping $\ell$ corresponds to a pure Nash equilibrium.

Since $G'$  is identical to $G$,  both maximum arity and modulo-treewidth remain unchanged.
\end{proof}

\fullver{
\begin{proof}[Proof of Theorem \ref{thm:tractable}]
%$2\rightarrow 1$ (Tractability):\\
Let $\Gamma=(G, \{ U_i \}_{i\in N})$, where $G\in C$, be a colored hypergraphical game. We first construct $\red(\Gamma)$, which has a pure equilibrium if and only if $\Gamma$ has one. Then by the reduction in Lemma \ref{lem:tractable} we can obtain an instance $(G_2, H_2)$  of the homomorphism problem in which $G_2$ has bounded degree and bounded modulo-treewidth and whose answer is equivalent to whether $\red(\Gamma)$ has a PSNE. By Theorem~\ref{thm:grohe},
$(G_2, H_2)$ can be solved in polynomial time. We thus have
a polynomial-time algorithm for $\purechg(C,-)$.
\end{proof}
}%end fullver

\subsection{Hardness for graphical games}\label{sec:hardness}

%As a result of previous sections we can apply the  breakthrough result of Grohe \cite{grohe} (Theorem \ref{thm:grohe}) to the $\puregg$
%problem, and prove Theorem \ref{thm:complexity}.

%\begin{theorem}\cite{grohe}
%
%Assume $\textit{FPT}\neq \textit{W[1]}$. Then for every recursively enumerable class $C$ of colored hypergraphs with bounded edge size the following two statements are equivalent.
%
%\begin{enumerate}
%
%\item $\textit{HOM}^c(C, -)$ is in polynomial time.
%
%%\item $\textit{p-HOM}(C, -)$ is fixed-parameter tractable.
%
%\item $C$ has bounded modulo-treewidth.
%
%\end{enumerate}
%
%\label{thm:grohe}
%
%\end{theorem}

%
%\begin{theorem}
%
%Assume $\textit{FPT}\neq \textit{W[1]}$. Then for every recursively enumerable class $C$ of irreducible digraphs with bounded in-degree the following two statements are equivalent.
%
%\begin{enumerate}
%\item $\textit{pure-GG}(C, -)$ is in polynomial time.
%%\item $\textit{p-pure-GG}(C, -)$ is fixed-parameter tractable.
%\item %$C$ has bounded modulo-treewidth, i.e.\
%for every $G\in C$, $\hyper(G)$ has bounded modulo-treewidth.
%\end{enumerate}
%
%\end{theorem}
We first consider the hardness result for graphical games.
In Section \ref{sec:chghardness} we extend our approach to
colored hypergraph games.
\fullver{\begin{theorem}\label{thm:gghardness}
Assume $\textit{FPT}\neq \textit{W[1]}$. Then for every recursively enumerable class $C$ of irreducible digraphs with bounded in-degree,
if $\ppuregg(C,-)$ is fixed-parameter tractable, then
for every $G\in C$, $G$ has bounded treewidth.
\end{theorem}
}%end fullver

\fullver{First some remarks about our proof.
Usually, hardness reductions have the following outline: given an arbitrary instance of a problem that's known to be hard, construct in polynomial time an instance of the problem of interest (\puregg in our case).
NP-hardness reductions, including Gottlob \emph{et al.}'s proof \cite{Gottlob03} of NP-hardness of \puregg, have such structure.
In Appendix \ref{sec:directreduction}, we give a similar reduction that given an arbitrary instance of homomorphism problem, construct a graphical game such that there exists a homomorphism if and only if the graphical game has a pure equilibrium. Furthermore, modulo-treewidth of the left-hand graph is preserved.
Using this reduction, we can prove Theorem \ref{}, which considers $C_f$, the set of all bounded-indegree digraphs $G$ whose modulo-treewidth is at most $f(|V(G)|)$, and states that If $\puregg(C_{f}, -)$ is in polynomial time, then $f$ is a bounded function.
Although a novel result, this is weaker than the one we want to prove (Theorem \ref{thm:gghardness}).
Theorem \ref{thm:gghardness} characterizes the complexity of $\puregg(C, -)$ for
 an \emph{arbitrary} class $C$ of bounded-indegree irreducible digraphs.
The proof of Theorem \ref{} constructs graphical games with a particular structure, which is not always present in an arbitrary class $C$. Thus the approach in Appendix \ref{sec:directreduction}
is not sufficient in proving the hardness direction of Theorem \ref{thm:complexity}. }%end fullver

\shortver{As mentioned in the introduction, applying the hardness proof approach of \cite{Gottlob03} to our setting would create graphical games with a particular structure, which is not sufficient for our purpose because we want to characterized the complexity of $\puregg(C, -)$ given an \emph{arbitrary} class $C$.}%end shortver
We thus use a different construction in our proof of the hardness direction, which
starts with an arbitrary class $C$ of irreducible digraphs, constructs a bijective mapping to a class $C'$ of colored hypergraphs,
and then show that for any instance $(G,H)$ of $HOM(C',-)$ we can construct an equivalent
instance of $\puregg(C,-)$. We can then apply Theorem \ref{thm:grohe} to get the hardness result.

The key step of the proof is the following lemma.
Recall that given digraph $G$, $\hyper(G)$ is the colored hypergraph with hyperedge  $e_i$ (the edge that corresponds to vertex $i$ and its neighbors) being colored with color $i$.
\begin{lemma}\label{lem:gghardness}
Let $G$ be an irreducible digraph. Then for any colored hypergraph $H$,
there exists a graphical game $GG =$ \\
$\left(G, \{ U_i \}_{i\in N}\right)$ such that there is a
homomorphism from $\hyper(G)$ to $H$ if and only if $GG$ has a PSNE.
\end{lemma}

The reduction is outlined as follows. (We give a detailed proof of the lemma in Appendix \ref{sec:hardnessproof}.)
%For an arbitrary digraph $G\in C$, we use $\hyper(G)$ as the left-hand side of the homomorphism problem.
%Then given an arbitrary colored hypergraph $H$, construct a graphical game $\Gamma=(G,{U_i})$ such that there exists a homomorphism from  $\hyper(G)$ to $H$ if and only if $\Gamma$ has a pure equilibrium.
Each player's action set consists of $V(H)$ plus some ``failure actions'', in this case $T$ and $B$.
We define the utility for $i$, given a local strategy profile over $\nei(i)$, such that if the local strategy profile correspond to a hyperedge in $H$ of color $i$, then $i$ gets a high payoff (say 100), such that if there exists a homomorphism from  $\hyper(G)$ to $H$, then the corresponding strategy profile is a PSNE.

If the local strategy profile does not correspond to an edge of right color in $H$, we set the utilities so that player $i$ is forced to play one of the failure actions. This implies that out-neighbors of i are forced to play failure actions, and so on.
Now we just need to set utilities such that if at least one player plays failure actions, then no PSNE exists.
Recall that if $G$ was a DAG, then there always exists a PSNE; i.e. a game construction with no PSNE must contain
a cycle.
Fortunately $G$ is assumed to be irreducible,
which means that all of its terminal SCCs has a directed cycle of length at least 2. For each of the terminal SCCs, fix a cycle and set the utilities  of players on that cycle (given failure actions of their neighbors) to be a generalization of the Matching Pennies game: one of the players is incentivized to play the opposite failure action as his predecessor in the cycle, while all other players on the cycle are incentivized to imitate their predecessors.

If there is no homomorphism, then for any strategy profile, there must be one player forced to play failure actions, which implies that at least one terminal SCC play failure actions, which implies that one of these cycles are playing the generalized Matching Pennies game, which does not have a PSNE.

\shortver{Using Lemma \ref{lem:gghardness}, given an FPT algorithm for $\ppuregg(C,-)$ we can construct an FPT algorithm for $\textit{p-HOM}(C',-)$ where $C' = \{\hyper(G) | G\in \red(C)\}$.
This implies the hardness direction for graphical games.
}%end shortver

\fullver{
\begin{proof}[Proof of Theorem \ref{thm:gghardness}]

Suppose $\ppuregg(C,-)$ is fixed-parameter tractable.
By Lemma \ref{lem:fpt}, this implies that
$\ppuregg(\red(C),-)$ is fixed-parameter tractable.

Now we claim that $\textit{p-HOM}(C',-)$, where $C' = \{\hyper(G) | G\in \red(C)\}$, is in FPT.
This is because given a homomorphism problem instance $(\hyper(G), H)$ where $\hyper(G) \in C'$, Lemma \ref{lem:gghardness} states that the instance
can be answered by deciding the existence of PSNE of a game whose graph is $G$. Furthermore
we can uniquely determine the corresponding digraph $G\in \red(C)$
and then construct the graphical game as above in polynomial time. %\footnote{This is why we color $\hyper{G}$ in such a way:i.e.\ we have to ensure that the mapping from $G $ to $\hyper{G}$ is a bijection.}
%The resulting graphical game can be solved in poly time
Thus this is a valid fpt-reduction from $\textit{p-HOM}(C',-)$ to $\ppuregg(\red(C),-)$.
Since $G\in \red(C)$ and $\ppuregg(\red(C),-)$ is in FPT,
we have a FPT algorithm for $\textit{p-HOM}(C',-)$.

Hence, according to
Theorem~\ref{thm:grohe}, $C'$ must have bounded modulo-treewidth
which means  $\red(C)$ have bounded treewidth (by Lemmas \ref{lem:modtw} and \ref{lem:hypertw}).
\end{proof}
}%end fullver

\subsection{Hardness for colored hypergraphical games}\label{sec:chghardness}
To prove the hardness direction of Theorem \ref{thm:complexity}, it is sufficient to extend Lemma \ref{lem:gghardness} to colored hypergraphical games:
\begin{lemma}\label{lem:chghardness}
Let $G\in \Sigma$ be an irreducible colored hypergraph. Then for any colored hypergraph $H$,
there exists a colored hypergraphical game $\Gamma = (G, \{ U_c \}_{c\in \tau})$ such that there is a
homomorphism from $G$ to $H$ if and only if $\Gamma$ has a PSNE.
\end{lemma}

We sketch a proof of the lemma in this section.
At a high level, the main difficulty when extending our proof of Lemma \ref{lem:gghardness} to colored hypergraphical games is that players with the same color must have the same utility function. Instead of being able to specify the utility function for each player in the graphical game case, now we need to define one utility function $U_c$ for each color $c$.
In fact, our generalized Matching Pennies construction for the graphical game case cannot be directly applied to colored hypergraphical game, and our proof of Lemma \ref{lem:chghardness} instead uses a different construction involving $2n+1$ failure actions for each player.

Part of the hardness proof for graphical games can be relatively easily adapted to colored hypergrahical games:
each player's action set still consists of $V(H)$ plus some failure actions (to be specified later).
We set the utility function $U_c$ so that if the input action tuple corresponds to a hyperedge of color $c$ in $H$, then the utility is 100. This will ensure that if there exists a homomorphism, then the corresponding strategy profile is a PSNE. This concludes the proof of the ``if'' direction of Lemma \ref{lem:chghardness}.

The ``only if'' direction is more difficult. In particular, it is difficult to define the utilities for the failure actions in a way that respects the color constraints.
For one, we would not be able to express the generalized Matching Pennies game now: in the worst case all players may have the same color.
Also, we cannot specify a cycle and then define utility functions on that cycle in a way that ignore all edges not in the cycle: this would also require player-specific utility functions.

Thus we want the utilities given failure actions to not depend on the player.
For the simple case of a single cycle, the following construction is sufficient. (For notational convenience, we only specify the best response function $BR$, which maps a tuple of actions of the neighbors to a single action as the best response. Given the $BR$ function the utilities can be defined straightforwardly.)
\begin{lemma}
Given a colored hypergraph $G$, whose induced digraph consists of just one cycle with length $n$, the following
colored hypergraphical game on $G$ does not have PSNE:
\begin{itemize}
\item each player's actions are the integers $0,\ldots,p-1$;
\item let $BR(a) = (a+1) \mod p$ where $p\geq n+1$.
\end{itemize}
\end{lemma}
We omit the straightforward proof. If we think of $BR$ as arcs from $a$ to $BR(a)$, then the digraph on actions form a $p$-cycle.

This can be extended to strongly connected digraphs, by the the following construction:
\begin{lemma}\label{lem:stronglyconnected}
Given a colored hypergraph $G$, whose induced digraph is strongly connected, the following
colored hypergraphical game on $G$ does not have PSNE:
\begin{itemize}
\item
each player's actions are the integers $0,\ldots,n$;
\item given neighbors' actions $( s_1,\ldots,s_m )$, let $BR( s_1,\ldots,s_m ) = (\max\{s_1,\ldots,s_m\} +1) \mod(n+1) $.
\end{itemize}
\end{lemma}
The intuition is that for each strategy profile, at least one neighbor is "activated" in the following sense:
Given digraph $G=(V,E)$, strat profile $\mathbf{s}$, we say an edge $(u,v)\in E$ is active if $$u\in \arg\max_{u': (u,v)\in E}    s_{u'},$$ i.e. $u$'s action under $\mathbf{s}$ is maximal among $v$'s neighbors.
Let $G'$ be the subgraph of $G$ where we only keep the active edges, i.e. for each player $i$,  only keep the edge from the neighbor playing the highest action among neighbors.
We claim that $G'$ must contain a cycle, i.e. is not a DAG.
This is because $G$ is strongly connected, which means it has  no source, i.e. all vertices of $G$ have positive number of incoming edges. This implies that all vertices of $G'$ have positive number of incoming edges, i.e. $G'$ has no source. Therefore $G'$ is not a DAG.

Since $G'$ must contain a cycle, on that cycle $BR(a)= (a+1)\mod(n+1)$, which implies that at least one player on that cycle is not playing a best response. Therefore $\mathbf{s}$ must not be a PSNE.

%-- extend: given a digraph with no sink
The above construction does not directly work for the general case of digraphs with no sinks: now $G'$ could be a DAG.
It turns out that we can indeed fix the construction to work for all digraphs with no sinks. At a high level, instead of forming a best-response cycle with the actions, we form a $\rho$ shape with a cycle and a tail.

We now complete the specification of the utility functions for our construction for Lemma \ref{lem:chghardness}.
The failure actions are $1,\ldots,2n+1$.
If the input action tuple of $U_c$ does not correspond to a hyperedge in $H$ with the same color $c$, then:
\begin{itemize}
\item if no neighbors are playing failure actions, then utility of playing failure action 1 is +1, all others -100;
\item otherwise, let $f_{\text{max}}$ be the max failure action of neighbors. If $f_{\text{max}}<2n+1$, then $BR=f_{\text{max}}+1$; otherwise (i.e. $f_{\text{max}}=2n+1$), let $BR=n+1$.
\end{itemize}

We now sketch the ``only if'' direction of Lemma \ref{lem:chghardness}. If there is no homomorphism, then for any strategy profile, some player must play failure actions, which implies that some SCC (and all SCCs reachable form there) must play failure; the other SCCs must not play failure actions.
Given a strategy profile, consider the ``earliest reached" non-singleton SCC, as defined by the following process:
 go through SCCs in topological order, in the direction of the edges; return the first non-singleton SCC whose players choose failure actions.
All earlier SCCs are either not playing failure actions, or a singleton SCC that is playing failure action $a<=n$.

Given strategy profile $\mathbf{s}$,
consider the graph $G_f$, which is the subgraph of $G$ restricted over the earliest non-singleton SCC and all earlier singleton SCCs playing failure actions:
%-- remove edges from non-max neighbors ->G'_f

First, we claim that if one player $j$ in the non-singleton SCC is playing an action less or equal to $n$, then $\mathbf{s}$ must not be a PSNE. This is because for such an action $b\leq n$ to be a best response, all incoming neighbors
within the SCC must be playing even lesser failure actions. If we iteratively follow an incoming neighbor within the SCC, thus with decreasing actions, we either encounter $j$ again, with action less than $b$, a contradiction, or a player $k$ playing action $1$. For 1 to be a BR, all neighbors must not play failure actions; but $k$ is in a non-singleton SCC with all players playing failure actions, so there must be at least one neighbor playing failure actions, a contradiction.

Therefore, in order for $\mathbf{s}$ to be an PSNE, all players in the non-singleton SCC must play failure actions greater than $n$. We claim that if this is the case, then each vertex in the non-trivial SCC must have an active neighbor in the same SCC. This is because
if an edge from a singleton SCC in $ G_f $ to a player $i$ the non-singleton SCC in $G_f$ is active, then because the player in the singleton SCC is playing some action $a<=n$, the target player $i$ in the non-singleton SCC must have an inactive neighbor in that SCC, i.e. some player $j$ in the non-singleton SCC is playing $b<a\leq n$. We have argued above that this would contradict with $\mathbf{s}$ being a PSNE.
Therefore, each vertex in the non-trivial SCC must have an active neighbor in the same SCC.
By the same argument as for the strongly connected digraph case, there must exist an active cycle in the SCC. Since every player on that active cycle is playing an action greater than $n$, they are playing a shifted version of the game in Lemma \ref{lem:stronglyconnected}.
This implies that $\mathbf{s}$ cannot be a PSNE.
%This concludes our proof sketch for Lemma \ref{lem:chghardness}.

%\subsection{Reducible Graphs}
\label{sec:reducible}

%end note

\section{Discussion and Future Work}
%In this paper, we completely characterize the complexity of finding pure Nash equilibrium in graphical games $(G, \{ U_i \}_{i\in N})$ based on the structure of $G$ with the assumption that $G$ has constant in-degree.

Our results can be understood as establishing an equivalence between  PSNE problems and homomorphism problems. Such a equivalence relation is closer than the kind of equivalence between two NP-complete problems:
we are in fact showing a family of equivalences, between $\purechg(C,-)$ for an arbitrary class $C$ and $\textit{HOM}(\red(C),-)$. On the other hand, our results also show that there are certain differences between
the two problems:  because
in a graphical/colored hypergraphical game each player has at least one best response regardless of her neighbors' actions,  we can iteratively remove sinks without affecting the answer,
whereas the same does not hold for homomorphism problems in general.

We have focused on the decision problem on the existence of pure Nash equilibria. Related problems include counting the number of pure Nash equilibria and finding one such equilibrium if one exists. On the homomorphism problem side, Dalmau and Jonsson \cite{Dalmau04counting} gave a characterization of the complexity of the counting version of
$\textit{HOM}(C,-)$, while the characterization for the construction problem is still open.
An interesting direction is to adapt our reductions to the counting and construction versions of these problems, as well as to cases with unbounded in-degree.
Another direction is to use similar techniques to prove  characterizations for other game representations such as action graph games~\cite{ActionGraph,JiangKLB06AGG}.

\begin{small}
\bibliographystyle{abbrv}
%\blibliographystyle{splncs}
\bibliography{AGGPure}
\end{small}

\appendix
\section{Proof of Lemma 13} \label{sec:fptproof}  %\ref{lem:fpt}
\begin{proof}
\fullver{\note{Should we write this as a formal fpt-reduction?} }
Given class $\mathcal{C}$ with in-degree bounded by $\indeg$, and an FPT algorithm for $\ppuregg(\mathcal{C},-)$, we now construct a fixed-parameter tractable algorithm for $\ppuregg(\red(\mathcal{C}),-)$.
Given $G' \in \red(\mathcal{C})$, we run the following algorithm:

\begin{enumerate}
\item Enumerate the class $\mathcal{C}$ until we find a graph $G\in \mathcal{C}$ such that $G'=\red(G)$.
\item Run algorithm for $\ppuregg(\mathcal{C},-)$ on $G$.
\end{enumerate}

We claim that Step 1 always terminates and its running time is bounded by a computable function on the size of $G'$.
This is because the class $C$ is recursively enumerable, and because
by definition, for each $G' \in \red(\mathcal{C})$,
there exists a graph $G\in \mathcal{C}$ such that $G'=\red(G)$.
%\note{It is easy to see that step 1 terminates in finite time;
%but how do we show that the running time is bounded by a computable
%function of $|G'|$? Is it true that if the step can expressed as
%a single Turing Machine, then its running time is bounded by a
%computable function?}
Therefore we have a fixed-parameter tractable algorithm
for $\ppuregg(\red(\mathcal{C}),-)$.
\end{proof}

\fullver{
\section{Proof of a weaker hardness result using direct reduction}\label{sec:directreduction}
%In this section we show that the homomorphism problem between hypergraphs is equivalent to the pure Nash equilibrium problem in graphical games. Moreover, we show that the reductions preserve edge size and modulo-treewidth bounds.

%\begin{theorem}
%The following problems are polynomial-time equivalent.
%\begin{enumerate}
%\item Finding homomorphism between colored hypergraphs.
%%\item Finding homomorphism between uncolored hypergraphs.
%\item Finding pure NE of graphical games.
%\end{enumerate}
%\label{thm:main}
%\end{theorem}
%

\begin{theorem}
Given an arbitrary instance $(G, H)$ of the homomorphism problem,
 a graphical game $\Gamma$ can be constructed in polynomial time such that there is a homomorphism from $G$ to $H$ if and only if $\Gamma$ has a PSNE.
\end{theorem}
\begin{proof}

%We show the equivalence by reducing problems to each other in the order $2\rightarrow 1\rightarrow 3\rightarrow 2$.

%$2\rightarrow 1$:\\

%It's trivial as $2$ is a special case of $1$.

Let $(G, H)$ be an instance of the  homomorphism problem. Here both $G$ and $H$ are colored hypergraphs. The homomorphism problem asks whether there exists a mapping $h$ from vertices of $G$ to $H$ such that for every edge $e=(v_1, v_2, \cdots, v_k)$ of $G$, $h(e) = (h(v_1), h(v_2), \cdots, h(v_k))$ is an edge in $H$  with the same color as $e$. We show how to solve this by finding  PSNE of some graphical games.

In the game that we construct there are $n = |V(G)|$ players corresponding to vertices of $G$ and each player chooses, as his action, a vertex of $H$ to map to. Each player gets a constant utility, say $1$, regardless of his action.

In addition we create two vertices $x(e)$ and $y(e)$ for every edge $e$ in $G$. The utility of both $x(e)$ and $y(e)$ depends on $x(e)$, $y(e)$ and all the vertices in $e$. The action of both $x(e)$ and $y(e)$ is either $g$ (for good) or $b$ (for bad). Their utilities are defined as follows. Let $C(e)$ be a boolean value indicating whether the actions of players in $e$ correspond to an edge with the same color in $H$ (i.e. a correct mapping for $e$).

\begin{itemize}

\item $x(e)$ gets $1$ iff the triple $(C(e), x(e), y(e))$ is either $(true, g, g)$ or $(false, g, g)$ or $(false, b, b)$ and gets zero otherwise.

\item $y(e)$ gets $1$ iff the triple $(C(e), x(e), y(e))$ is either $(true, g, g)$ or $(false, b, g)$ or $(false, g, b)$ and gets zero otherwise.

\end{itemize}

If $C(e)$ is false then $x(e)$ and $y(e)$ would have opposite interests and, hence, the game won't have any PSNE. Therefore, the game has a PSNE if and only if there exists a homomorphism from $G$ to $H$. Notice that if there is a proper mapping $h$ from $G$ to $H$ then the players can obviously play according to $h$ and have a PSNE.

\end{proof}

%\subsection{Tree-width modulo homomorphism}
In the following lemma we show that both maximum edge size and modulo-treewidth are preserved for the left-side hypergraph in the proof of Theorem~\ref{thm:main}.
%We first show that both degree and modulo-treewidth bounds are preserved under the reductions in Theorem~\ref{thm:main}

\begin{lemma}

In the proof of Theorem~\ref{thm:main} the maximum edge size and module-treewidth bounds are preserved.
\label{thm:treewidth}
\end{lemma}

\begin{proof}

Assume $G$ has maximum edge size $r$ and modulo-treewidth $w$.

Let $GG' = (G', \{U_i\}_{i\in N} )$ be the  graphical game obtained by reduction $1\rightarrow 2$ in the proof of Theorem~\ref{thm:main}. Notice that for every edge $e\in G$, $\hyper(G')$ has an edge $e' = e\cup \{x(e), y(e)\}$.
Let $G_1$ be a graph that is homomorphically equivalent to $G$ with  the minimum treewidth, $w$, and $(T, W)$ be its optimal tree-decomposition. Furthermore, let $h$ be a homomorphism from $G$ to $G_1$.

We build a hypergraph $G'_1$ with $$V(G'_1) = V(G_1) \bigcup_{e\in \hyper(G)} \{x(e), y(e)\}$$ and for every edge $e$ in $G$ we put an edge $e'$ in $G'_1$ whose vertices are $h(e)\cup \{x(e), y(e)\}$.

 It's clear that $G'_1$ is a homomorphism of $\hyper(G')$ ($x(e)$ maps to $x(e)$,  $y(e)$ maps to $y(e)$ and  the rest of vertices follow the mapping $h$). For every edge $e\in G$ vertices in $e$ make a clique in $\pri(G)$, and hence $h(e)$ is a clique in $\pri(G_1)$. So, by Lemma~\ref{lemma:kloks}, there must be a node $i$ in $T$ such that $h(e)\subset W_i$. Now we can attach a leaf $i'$ to $i$ and set $W_{i'} = W_i\cup \{x(e), y(e)\}$. Obviously this would be a tree-decomposition for $G'_1$ and its tree-width is at most $w+2$. It's also easy to verify that the maximum edge size  of $\hyper(G')$ is $r+2$.
\end{proof}

Using the above reduction, we can show the following, which is a
novel result but is weaker than the hardness result proved in
Section \ref{sec:hardness}:
\begin{theorem}
Assume $\textit{FPT}\neq\textit{W[1]}$.
Let $f: \mathds{N}\rightarrow\mathds{N}$ be a non-decreasing function. Define
$C_{f}$ to be  the set of all bounded-indegree digraphs $G$ whose induced hypergraphs have
modulo-treewidth at most $f(|V(G)|)$.
If $\puregg(C_{f}, -)$ is in polynomial time, then $f$ is a bounded function.
\end{theorem}
\begin{proof}
Suppose $\puregg(C_{f}, -)$ is in polynomial time. We construct a polynomial time algorithm for $\textit{HOM}(\hyper(C_{f'}), -)$, where $f'(x)\equiv f(x)-2$, as follows.
Given an instance $(G,H)$ of $\textit{HOM}(\hyper(C_{f'}), -)$,
use the above reduction to create a graphical game $(G',\{U_i\})$ with
modulo-treewidth bounded by $f'(|V(G)|)+2$. Since $|V(G)|<|V(G')|$,
this is bounded by $f(V(G')$. We can then apply the polynomial
time algorithm for $\puregg(C_{f}, -)$.
Therefore $\textit{HOM}(\hyper(C_{f'}), -)$ is in polynomial time, which then implies
that $\hyper(C_{f'})$ has bounded modulo-treewidth, due to Theorem \ref{thm:grohe}.
This implies that $f$ is a bounded function.
\end{proof}

}%end fullver

\section{Proof of Lemma 19}\label{sec:hardnessproof} %\ref{lem:gghardness}

\begin{proof}
Let $G$'s terminal SCCs be $\pi_1,\ldots,\pi_m$.
For each $\pi_j$,
fix a cycle $i^j_{0}\rightarrow\ldots\rightarrow i^j_{r_j}\rightarrow i^j_{0}$. This is always possible since $G$ is
irreducible and, hence, every terminal SCC has at least two vertices. These cycles are disjoint since the terminal SCCs are maximal SCCs.

The graphical game $GG$ is constructed as follows.
\begin{itemize}
\item Each player $i$'s action set is $V(H) \cup T \cup B$. $T$ and $B$ are the ``failure actions''.

\item Player $i$'s utility: by the definition of graphical game, his utility depends on the actions chosen by him and his neighbors.
Let $p_i\in \prod_{j\in  \nei(i)} S_j$ be the tuple of
actions chosen by $i$ and its neighbors.
\begin{enumerate}
\item If $p_i$ corresponds to a hyperedge in H with the same color as $e_i$ (the edge corresponding to $\nei(i)$), then $i$'s utility is 100.
\item Otherwise:
\begin{enumerate}
\item If $i$ is not playing one of the failure actions $T$ or $B$, then $i$ gets $-100$.

%\item if there are more than one player playing failure actions in $p_i$, and $i$ is the one with the smallest index, and $j$ is the second smallest, then $i$ gets 1 if its action matches $j$, and $-1$ otherwise.

%\item if there are more than one player playing failure actions in $p_i$, and $i$ is the one with the second smallest index, and $j$ is the smallest, then $i$ gets $-1$ if its action matches $j$, and $1$ otherwise.

\item Otherwise, if $i = i^j_{k}$ in one of the pre-defined cycles:
\begin{enumerate}
\item If $k>0$, then $i=i^j_{k}$'s payoff depends only on the actions
of herself and $i^j_{k-1}$. If $i^j_{k-1}$ plays other than $T$ or
$B$, $i$ gets 0 by playing either $T$ or $B$. Otherwise, $i$ gets 1
if both she and $i^j_{k-1}$ plays $T$ or both plays $B$, and -1
otherwise.
\item If $k=0$, then $i=i^j_{0}$'s payoff depends only on the actions
of herself and $i^j_{r_j}$. If $i^j_{r_j}$ plays other than $T$ or
$B$, $i$ gets 0 by playing either $T$ or $B$. Otherwise, $i$ gets -1
if both she and $i^j_{r_j}$ plays $T$ or both plays $B$, and 1
otherwise.
\end{enumerate}
\item Otherwise, $i$ gets 0 (by playing either $T$ or $B$).
\end{enumerate}
\end{enumerate}
\end{itemize}
We claim that this graphical game has a PSNE if and only if there
is a homomorphism from $\hyper{G}$ to $H$.

{\bf if part}: if there exists a homomorphism from $\hyper{G}$ to $H$ whose mapping function is $h$, then in the graphical game, the strategy in
which each player $i$ plays $h(i)$ is a Nash equilibrium.

{\bf only if part}: A PSNE where everyone gets 100 corresponds to
a homomorphism. Furthermore, the only PSNE of the graphical game
are ones where every player gets 100. This is because if some player
$i$ fails to get 100, then he has to play $T$ or $B$ to avoid the
-100 penalty. This makes all his (outgoing) neighbors fail to get 100 as well,
so they also have to play $T$ or $B$. $i$ is either part of a terminal $\pi_j$ or there is a path to a vertex in a terminal $\pi_j$. Since $i$ plays failure actions all players in $\pi_j$ must play failure actions as well.
Therefore the pre-defined cycle $i^j_{0}\rightarrow\ldots\rightarrow i^j_{r_j}\rightarrow i^j_{0}$
play failure actions. The utilities are set up so that the players
in this cycle are playing a game similar to Matching Pennies, and it
is straightforward to verify that there is no PSNE if they play
the failure actions. Therefore there's no PSNE unless everyone
gets 100.
\end{proof}

%\section{Proof Sketch of Lemma \ref{lem:chghardness}}\label{sec:chgproof}

\end{document}